\documentclass[final, 11pt]{article}

\usepackage{ltexpprt} \usepackage{amssymb} \usepackage{amsmath} \usepackage{appendix}

\begin{document}

\title{Packet Scheduling in a Size-Bounded Buffer}

\author{Fei Li~\thanks{Department of Computer Science, George Mason University, Fairfax, VA 22030, USA. {\tt lifei@cs.gmu.edu}}}

\maketitle


\begin{abstract}
We consider algorithms to schedule packets with values and deadlines in a size-bounded buffer. At any time, the buffer can store at most $B$ packets. Packets arrive over time. Each packet has a non-negative value and an integer deadline. In each time step, at most one packet can be sent. Packets can be dropped at any time before they are sent. The objective is to maximize the total value gained by delivering packets no later than their respective deadlines. This model generalizes the well-studied {\em bounded-delay model} (Hajek. CISS 2001. Kesselman et al. STOC 2001). We first provide an optimal offline algorithm for this model. Then we present an alternative proof of the $2$-competitive deterministic online algorithm (Fung. arXiv July 2009). We also prove that the lower bound of competitive ratio of a family of (deterministic and randomized) algorithms is $2 - 1 / B$.
\end{abstract}


\section{Introduction}

In this paper, we study a model called the {\em bounded buffer model} for Quality-of-Service (QoS) buffer management. Time is discrete. Packets arrive over time. Each packet $p$ has a non-negative value $v_p \in \mathbb R^+$ and an integer deadline $d_p \in \mathbb Z^+$. The deadline $d_p$ specifies the time by which $p$ should be sent. There is a buffer with a limited size of $B \in \mathbb Z^+$. At any time, the buffer can store no more than $B$ packets. Packets already existing in the buffers can be dropped at any time before they are served. A dropped packet cannot be delivered any more. In each time step, at most one packet from the buffer can be sent. The objective is to maximize {\em weighted throughput}, defined as the total value of the transmitted packets by their respective deadlines.

We design both offline and online algorithms for this model.  We use competitive ratio to measure an online algorithm's performance versus an optimal clairvoyant algorithm who knows the whole input in advance. A deterministic (randomized) online algorithm is called {\em $k$-competitive} if its (expected) weighted throughput on {\em any} finite instance is at least $1 / k$ of the weighted throughput of an optimal offline algorithm on this instance. $k$ is known as the online algorithm's {\em competitive ratio}~\cite{BY98}.


\subsection{Related work.}

Since the {\em bounded-delay model} for QoS buffer management was introduced in~\cite{KLMPSS04, H01}, many researchers have considered this model as well as its variants~\cite{KLMPSS04, H01, CF03, CCFJST06, CJST07, LSS05, LSS07, EW07}. Most research results assume that the buffer has sufficient space to accommodate all released packets before they are delivered or they expire. Instead, we consider a more practical model in this paper. In our model, the buffer has a finite size $B \in \mathbb Z^+$ such that at any time, no more than $B$ packets can be concurrently staying in the buffer. The bounded buffer model generalizes the bounded-delay model, given that we allow the buffer size $B$ to be larger than any packet's {\em slack time}. (A packet's slack time is defined as the difference between its deadline and release time.)

For the bounded-delay model, an optimal offline algorithm has been proposed in~\cite{KLMPSS04}, running in $O(n \log n)$ time where $n$ is the number of packets released. For online algorithms, the best known lower bound of competitive ratio of deterministic algorithms is $\phi = (1 + \sqrt{5}) / 2 \approx 1.618$~\cite{H01, CF03, AMZ03}; this lower-bound also applies to instances in which the deadlines of the packets (weakly) increase with their release dates. A simple greedy algorithm that always schedules the maximum-value packet in the buffer is $2$-competitive~\cite{H01, KLMPSS04}.  A generalization of the greedy algorithm, which always schedules the earliest packet with a value at least $1 / \alpha$ ($\alpha \ge 1$) times of the maximum-value of a packet~\cite{CCFJST06}, has a competitive ratio of asymptotically $2$.  Chrobak et al.~\cite{CJST07} discuss a modification with one status bit that results in an algorithm with a competitive ratio of $64 / 33 \approx 1.939$. For a variant in which the deadlines of the packets (weakly) increase with their release dates, Li et al.~\cite{LSS05} propose an optimal deterministic $\phi$-competitive algorithm. Using the same analysis, but in a more complicated way, Li et al. provide a ($3 / \phi \approx 1.854$)-competitive deterministic algorithm~\cite{LSS05} for the general model. Independently, Englert and Westermann present a $1.894$-competitive deterministic memoryless algorithm and a ($2 \sqrt{2} - 1 \approx 1.828$)-competitive deterministic algorithm~\cite{EW07}. Closing the gap $[1.618, \ 1.828]$ of competitive ratio for deterministic algorithms is a difficult open problem. A randomized online algorithm with a competitive ratio of $e / (e - 1) \approx 1.582$ is proposed in~\cite{CCFJST06}. The lower bound of competitive ratio of randomized algorithms is $1.25$. How to tighten the gap $[1.25, \ 1.582]$ in the randomized bounded-delay model remains open. Recently, an algorithm achieving a competitive ratio of $4 / 3 \approx 1.33$ against an oblivious adversary has been proposed in~\cite{JLSS09} for instances in which packet deadlines weakly increase with their release time.

The bounded buffer model is studied by Li~\cite{L09}. Its generalization, called the {\em multi-buffer model}, is considered by Azar and Levy~\cite{AL06}. In~\cite{L09}, a $3$-competitive deterministic algorithm and a ($\phi^2 \approx 2.618$)-competitive randomized algorithms are given. Fung~\cite{F09} provides a $2$-competitive deterministic algorithm and in this paper, we present an alternative proof. Azar and Levy~\cite{AL06} provide a $9.82$-competitive deterministic algorithm, which also works for the multi-buffer model. For the multi-buffer model, Li~\cite{L09b} improves the competitive ratio to $3 + \sqrt{3} \approx 4.723$.


\subsection{Our contributions.}

The paper is organized as follows. In Section~\ref{sec:offline}, we present an optimal offline algorithm for the bounded buffer model. Then in Section~\ref{sec:online}, we provide an alternative proof of the $2$-competitive algorithm given by Fung, as well as the lower bound $2 - 1 / B$ for a broad family of online algorithms.


\section{An Offline Algorithm}
\label{sec:offline}

We define a canonical order of delivering packets: If two packets are both in the buffer and are to be delivered by the specified algorithm, the one with an earlier deadline is sent. We call this order EDF (Earliest-Deadline-First).

\begin{theorem}
For the bounded buffer model, there exists an optimal offline algorithm running in $O(n^2)$ time, where $n$ is the number of packets released.
\end{theorem}

\begin{proof}
Fix an input sequence $\mathcal I$. We sort all packets in $\mathcal I$ in non-increasing value order, with ties broken in favor of the one with a later deadline. We start from a set of packets $S = \emptyset$. For each packet $j \in ({\mathcal I} \setminus S)$, we pick up $j$ in order and run EDF to examine whether all packets in $S \cup \{j\}$ can be delivered successfully by their respective deadlines. If ``yes'', we update $S$ with $S \cup \{j\}$. For each examined packet $j$, no matter whether we insert $j$ into $S$ or not, we drop it out of $\mathcal I$. We examine all packets in $\mathcal I$ in order till $\mathcal I$ gets empty. We claim that the schedule on $S$ we finally have is optimal, based on the matroid property of this model.

Let $|{\mathcal I}| = n$. Sorting packets in $\mathcal I$ takes $O(n \log n)$ time. The buffer has at most $B$ packets at any time, thus, each packet insertion (in increasing deadline order) takes $O(\log B)$ time. Running EDF over a set of packets $S \cup \{j\}$ takes time $|S| + 1 + \log B \le n + \log B$. Thus, the total running time of the algorithm is $O(n \log n + n (n + \log B)) = O(n^2 + n \log B)$. Note $B < n$ (otherwise, this model is the bounded-delay model and has a running time of $O(n \log n)$~\cite{KLMPSS04}), thus, our algorithm has a running time of $O(n^2)$. $\Box$
\end{proof}

\begin{corollary}
For the variant in which all packets are with the same value, an online algorithm EDF is optimal with a running time of $O(n \log B)$.
\end{corollary}

\begin{proof}
If all packets are with the same value, our objective is to maximize the number of packets delivered by their deadlines. Thus, we simply send packets using the policy EDF. For each packet, inserting it in the buffer (in increasing deadline order) takes $O(\log B)$ time. Thus, the total running time of EDF is $O(n + n \log B) = O(n \log B)$. $\Box$
\end{proof}


\section{Online Algorithms}
\label{sec:online}

At first, we introduce a few concepts. Then we prove the lower bound $2 - 1 / B$ of competitive ratio for a broad family of deterministic and randomized online algorithms. At last, we present a deterministic $2$-competitive online algorithm and its analysis.

\begin{Definition}
{\bf Provisional schedule~\cite{CJST07a, EW07}}. At any time $t$, a {\em provisional schedule} is a schedule for the pending packets at time $t$ (assuming no future arriving packets). This schedule specifies the set of packets to be transmitted, and for each it specifies the delivery time.
\end{Definition}

\begin{Definition}
{\bf Optimal provisional schedule~\cite{CJST07a, EW07}}. Given a set of pending packets, an {\em optimal provisional schedule} is the one achieving the maximum total value of packets among all provisional schedules.
\end{Definition}

We use ${\bf S}_t$ to denote both the provisional schedule for time steps $[t, \ +\infty)$ and the set of packets delivered successfully in the schedule. All known online algorithms for the bounded-delay model~\cite{LSS05, CJST07, LSS07, EW07} calculate the optimal provisional schedules at the beginning of each time step. These algorithms differ only by the packets they select to send. The online algorithms in such a broad family are defined as the {\em best-effort admission algorithms}.

\begin{Definition}
{\bf Best-effort admission algorithm}. Consider an online algorithm ON and a set of pending packets ${\bf P}_t$ at any time $t$. If ON calculates the optimal provisional schedule ${\bf S}_t$ on ${\bf P}_t$ and selects one packet from ${\bf S}_t$ to send in each step $t$, we call ON a {\em best-effort admission algorithm}.
\end{Definition}


\subsection{The lower bound of competitive ratio.}

In this section, we create an instance to prove that the lower bound of competitive ratio for all best-effort admission algorithms is $2 - 1 / B$. Note that for the bounded-delay model, the buffer size $B$ is implicitly specified by $+\infty$ and the lower bound of competitive ratio is $\phi$~\cite{AMRR05, H01, CF03}. This lower bound holds for the best-effort admission algorithms as well. Here, we improve the lower bound from $\phi$ to $\max\{\phi, \ 2 - 1 / B\}$ for the bounded buffer model in which the buffer size is restricted by $B$.

\begin{theorem}
For the bounded buffer model, the lower bound of competitive ratio for the best-effort admission algorithms is $\max\{\phi, \ 2 - 1 / B\}$, where $B$ is the buffer size.
\label{theorem:lowerbound}
\end{theorem}

\begin{proof}
In the following instance, we will show: {\em If the buffer size is bounded, the packets that the optimal offline algorithm chooses to send may not be from the optimal provisional schedule calculated by the online algorithm, even if both algorithms have the same set of pending packets}. This property does not hold in the bounded-delay model; and it leads that any deterministic best-effort admission algorithm cannot achieve a competitive ratio better than $2$.

Assume the buffer size is $B$. Let a best-effort admission online algorithm be ON. We use $(v_p, \ d_p)$ to represent a packet $p$ with a value $v_p$ and a deadline $d_p$. Initially, the buffer is empty.  A set of packets, from which the optimal offline algorithm will accept $b - 1$ packets from them and eventually send, are released: $(1, \ B + 1), \ (1, \ B + 2), \ \ldots, \ (1, \ B + B)$. Notice that all packets released have deadlines larger than the buffer size $B$. The optimal offline algorithm drops $(1, \ B + 1)$, and keeps $(1, \ B + 2), \ \ldots, \ (1, \ B + B)$ in its buffer. In the same time step, $B$ packets $(1 + \epsilon, \ 1), \ (1 + \epsilon, \ 2), \ \ldots, \ (1 + \epsilon, \ B)$ are released afterwards. There are no more new packets arriving in this step. The optimal offline algorithm only accepts $(1 + \epsilon, \ 1)$. Thus, after processing arrivals in step $1$, the optimal offline algorithm send the packet $(1 + \epsilon, \ 1)$. Instead, ON calculates the optimal provisional schedule in step $1$ which includes all these newly arriving packets with value $1 + \epsilon$. All such packets will be accepted by ON, but the packets $(1, \ B + i)$, $\forall i = 1, \ 2, \ \ldots, \ B$, will be dropped due to the buffer size constraint. ON sends a packet with value $1 + \epsilon$ in the first step.

At the beginning of each step $i = 2, \ 3, \ \ldots, \ B$, only one packet $(1 + \epsilon, \ i)$ is released. At the end of step $B$, no new packets will be released in the future. Since the time after the first step, all packets available to ON have their deadlines $\le B$. Thus, ON cannot schedule sending packets with a total value $\ge (1 + \epsilon)  (B - 1)$ in the time steps $2, \ 3, \ \ldots, \ B$. Since there is one empty buffer slot at the beginning of each time step $i = 2, \ 3, \ \ldots, \ B$, the optimal offline algorithm can accept and send all newly released packets $(1 + \epsilon, \ i)$ in steps $i = 2, \ 3, \ \ldots, \ B$. At the end of step $B$, the packets $(1, \ B + 2), \ (1, \ B + 3), \ \ldots, \ (1, \ B + B)$ are still remained in the optimal offline algorithm's buffer (they are not in ON's buffer though). Since there is no future arrivals, these $b - 1$ packets will be transmitted eventually by the optimal algorithm in the following $b - 1$ steps. The total value of ON achieves is $(1 + \epsilon) B$ while the optimal offline algorithm gets a total value $(1 + \epsilon)  B + (B - 1)$. The competitive ratio for this instance is
\begin{equation}
c = \frac{(1 + \epsilon) B + (B - 1)}{(1 + \epsilon) B} = 2 - \frac{1 + B \cdot \epsilon}{b + B \cdot \epsilon} \ge 2 - \frac{2}{B}, \ \ \ \mbox{if } \epsilon \cdot B = 1 \mbox{ and } B \ge 2.
\label{eq:change}
\end{equation}

If $B$ is large, ON cannot perform asymptotically better than $2$-competitive. This lose is due to ON calculating an optimal provisional schedule to find out the packet to send in each time step. Theorem~\ref{theorem:lowerbound} is proved. $\Box$
\end{proof}


\subsection{A deterministic algorithm GRQ and its analysis.}
\label{sec:me}

The algorithm GRQ~\cite{F09} works as follows at each time $t$.

\begin{enumerate}
\item Align packets in the buffer slots in non-increasing value order. If a packet $j$ cannot be scheduled in the time slot $t + S[i] - 1$, where $S[i]$ is the first available buffer slot, drop $j$.

\item Send the first packet.
\end{enumerate}

\begin{theorem}
For the bounded buffer model, GRQ is $2$-competitive.
\label{theorem:sgd}
\end{theorem}

\begin{proof}
We prove Theorem~\ref{theorem:sgd} using a potential function method. Our proof is motivated by Kimbrel's simple proof~\cite{K04} for the $2$-competitive greedy algorithm for the FIFO buffer model.

We give some notation. Fix an input sequence of arriving packets. This input sequence can be regarded as a sequence of {\em packet arrival events} and {\em packet delivery events}. Then, in our analysis, if not mentioning, we use the subscript $t$ to denote an event, instead of a time step $t$. A single time step may involve more than one arrival events but only one delivery event. We use $Q^\text{ALG}_t$ to denote the algorithm ALG's buffer at time $t$. The buffer's slots are denoted as $S[1], \ S[2], \ \ldots, \ S[B]$. We use $j \in S[i]$ to denote that a packet $j$ is in the buffer slot $S[i]$ in GRQ's buffer. For GRQ, each buffer slot $S[i]$ corresponds to a time step GRQ sends a packet at time $t + S[i] - 1$ in the current provisional schedule.

Without loss of generality, we assume OPT only stores the packets it sends. We compare $Q^\text{GRQ}_t$ and $Q^\text{OPT}_t$. Define $X_t$ as the set of packets in OPT's buffer that should be sent by time $t + |Q^\text{GRQ}_t| - 1$; let $z(t) = |X_t|$; $1 \le z(t) \le B$. We define a potential
\begin{equation}
\Phi = 2 \sum_{j \in S[i], i \le z(t)} v_j + \sum_{j \in S[i], i > z(t)} v_j - \sum_{k \in X_t} v_k.
\label{eq:diff}
\end{equation}

Let $V_t$ and $W_t$ denote the values of the algorithms GRQ and OPT gain respectively in step $t$. We will show that at any time, the change of $2 V_t - W_t + \Delta \Phi$ is always non-negative. We will prove the following Equation~\ref{eq:potential} holds all the time, separately for the events of packet arrivals and deliveries. Thus, it yields Theorem~\ref{theorem:sgd}.

\begin{equation}
2 V_t + \Delta (2 \sum_{j \in S[i], i \le z(t)} v_j + \sum_{j \in S[i], i > z(t)} v_j) \ge W_t + \Delta (\sum_{k \in X_t} v_k).
\label{eq:potential}
\end{equation}

Initially, $\Phi$ is $0$. We study cases of packet arrivals and deliveries.


\subsubsection{Packet deliveries.}

$z(t)$ is reduced by $1$ if OPT has a packet to send. The right side of Inequality~\ref{eq:potential} is always $0$. The left side of Inequality~\ref{eq:potential} is $0$ if OPT sends a packet. If OPT sends nothing, the left side of Inequality~\ref{eq:potential} is strict positive when GRQ sends a packet or $0$ if GRQ sends nothing.


\subsubsection{Packet arrivals.}

Consider a new arrival $p$. $V_t$ and $W_t$ are $0$. If both algorithms reject $p$, $\Phi$ does not change at all. We consider the cases that at least one algorithm accepts $p$. We will show that $\Delta \Phi$ never goes to negative.

\begin{enumerate}
\item Assume OPT rejects $p$. (GRQ accepts $p$.)

The change of the third term of $\Phi$ of Inequality~\ref{eq:diff} is $0$. Note $z(t)$ is unchanged. If $p$ is accepted, it has more value than the packet $j$ evicted, if any. Otherwise, GRQ will keep $j$ instead of $p$ in the buffer. Both the first term and the second term will not decrease.

\item Assume OPT accepts $p$.

Assume $p$ is rejected by GRQ. Any packet stored in the buffer slot $S[i]$ with $1 \le i \le \max\{B, \ d_p - t + 1\}$ has a value $\ge v_p$. Otherwise, GRQ will use $p$ to replace that less-value packet. If $p$ is sent by OPT by time $t + |Q^\text{GRQ}_t| - 1$, the change of the third term of Inequality~\ref{eq:diff} is increased by $v_p$. We claim that $z(t)$ is increased by $1$ and before $p$ arrives, there is at least one packet in GRQ's buffer in a buffer position beyond $z(t)$ with a value $\ge v_p$. If not, either OPT reject $p$ or GRQ accepts $p$. We increase $z(t)$ by $1$ and the first term of $\Phi$ of Inequality~\ref{eq:diff} is increased by at least $v_p$. The second term of $\Phi$ of Inequality~\ref{eq:diff} is not changed. If $p$ is sent by OPT later than $t + |Q^\text{GRQ}_t| - 1$, nothing is changed for Inequality~\ref{eq:diff}.

Assume $p$ is accepted by GRQ. Any packet in GRQ's buffer in a buffer slot $S[i]$ with $1 \le i \le \max\{B, \ d_p - t + 1\}$ has a not-less-value packet than the one it previously stored. We apply the same analysis as above and conclude that Inequality~\ref{eq:potential} holds.
\end{enumerate}


Based on our case study above, Theorem~\ref{theorem:sgd} is proved. $\Box$
\end{proof}


\bibliographystyle{plain}
\bibliography{buffer}


\end{document}